\documentclass[english,3p,twocolumn,superscriptaddress,prl]{revtex4-2}
\usepackage[toc,page]{appendix}
\usepackage[latin9]{inputenc}
\usepackage{amsmath}
\usepackage{amssymb, enumerate}
\usepackage{color}
\usepackage{babel}
\usepackage{graphicx}
\usepackage{epstopdf}
\usepackage{float}
\usepackage{thmtools}
\usepackage{thm-restate}
\usepackage{hyperref}
\usepackage{cleveref}
\usepackage{enumerate}
\usepackage{lipsum}
\usepackage{fixmath}
\usepackage{dsfont}
\usepackage{tikz}
\usepackage{pifont}
\usepackage{amsmath,amsfonts,amsthm}
\usepackage{graphicx}

\usepackage{pifont}

\newcommand{\ket}[1]{ | \, #1 \rangle} \newcommand{\bra}[1]{ \langle #1 \, |} 
\newcommand{\proj}[1]{\ket{#1}\bra{#1}} 
\newcommand{\kb}[2]{\ket{#1}\bra{#2}}

\DeclareMathOperator{\Tr}{Tr}

\DeclareMathOperator{\Tsn}{Tsn}
\DeclareMathOperator{\Boxes}{Boxes}

\DeclareMathOperator{\Flex}{Flex}
\DeclareMathOperator{\SWAP}{SWAP}
\DeclareMathOperator{\Op}{Op}
\DeclareMathOperator{\Spec}{Spec}
\DeclareMathOperator{\Eig}{Eig}
\DeclareMathOperator{\PSD}{PSD}
\DeclareMathOperator{\Dim}{Dim}
\DeclareMathOperator{\DPS}{DPS}

\usetikzlibrary{arrows.meta,bending,matrix,positioning}
\makeatletter

\@ifundefined{textcolor}{}
{%
 \definecolor{BLACK}{gray}{0}
 \definecolor{WHITE}{gray}{1}
 \definecolor{RED}{rgb}{1,0,0}
 \definecolor{GREEN}{rgb}{0,1,0}
 \definecolor{BLUE}{rgb}{0,0,1}
 \definecolor{CYAN}{cmyk}{1,0,0,0}
 \definecolor{MAGENTA}{cmyk}{0,1,0,0}
 \definecolor{YELLOW}{cmyk}{0,0,1,0}
}

\newtheorem{theorem}{Theorem}

\newtheorem{remark}[theorem]{Remark}
\newtheorem{observation}[theorem]{Observation}
\newtheorem{lem}[theorem]{Lemma}
\newtheorem{definition}[theorem]{Definition}

\makeatother

\begin{document}

\title{Hybrid no-signaling-quantum correlations}
\author{Micha{\l} Banacki}
\affiliation{International Centre for Theory of Quantum Technologies, University of Gda\'{n}sk, Wita Stwosza 63, 80-308 Gda\'{n}sk, Poland}
\affiliation{Institute of Theoretical Physics and Astrophysics, Faculty of Mathematics, Physics and Informatics, University of Gda\'{n}sk, Wita Stwosza 57, 80-308 Gda\'{n}sk, Poland}
\author{Piotr Mironowicz}
\affiliation{International Centre for Theory of Quantum Technologies, University of Gda\'{n}sk, Wita Stwosza 63, 80-308 Gda\'{n}sk, Poland}
\affiliation{Department of Algorithms and System Modeling, Faculty of Electronics, Telecommunications and Informatics, Gda\'{n}sk University of Technology, Gabriela Narutowicza 11/12, 80-233 Gda\'{n}sk, Poland} 
\affiliation{Department of Physics, Stockholm University, S-10691 Stockholm, Sweden} 
\author{Ravishankar Ramanathan}
\affiliation{Department of Computer Science, The University of Hong Kong, Pokfulam Road, Hong Kong}
\author{Pawe{\l} Horodecki}
\affiliation{International Centre for Theory of Quantum Technologies, University of Gda\'{n}sk, Wita Stwosza 63, 80-308 Gda\'{n}sk, Poland}
\affiliation{Faculty of Applied Physics and Mathematics, Gda\'{n}sk University of Technology, Gabriela Narutowicza 11/12, 80-233 Gda\'{n}sk, Poland} 

\begin{abstract}
Fundamental investigations in non-locality have shown that while the no-signaling principle alone is not sufficient to single out the set of quantum non-local correlations, local quantum mechanics and no-signaling together exactly reproduce the set of quantum correlations in the two-party Bell scenario. Here, we introduce and study an intermediate hybrid no-signaling quantum set of non-local correlations that we term $\textbf{HNSQ}$ in the multi-party Bell scenario where some subsystems are locally quantum while the remaining subsystems are only constrained by the no-signaling principle. Specifically, the set $\textbf{HNSQ}$ is a super-quantum set of correlations derived from no-signaling assemblages by performing quantum measurements on the trusted subsystems. As a tool for optimization over the set $\textbf{HNSQ}$, we introduce an outer hierarchy of semi-definite programming approximations to the set following an approach put forward by Doherty-Parillo-Spedalieri. We then show that in contrast to the set $\textbf{NS}$ of no-signaling correlations, there exist extreme points of $\textbf{HNSQ}$ in the tripartite Bell scenario that admit quantum realization. We perform an extensive numerical analysis of the maximal violation of the facet Bell inequalities in the three-party binary input-output scenario and study the corresponding self-testing properties. In contrast to the usual no-signaling correlations, the new set allows for simple security proofs of (one-sided)-device-independent applications against super-quantum adversaries.
\end{abstract}



\maketitle

\noindent\textit{Introduction.-} Quantum non-local correlations violating Bell inequalities \cite{Bell,RMPBellnonlocality} are of great fundamental interest besides giving rise to the powerful application of device-independent (DI) cryptography \cite{Barrett2005, PironioNature, PRLAcin, IEEEKessler, BRGH+16, RBHH+16}. In investigations of quantum non-locality, it has been very fruitful to study the correlations from the outside by investigating general no-signaling correlations constrained only by the principle of no superluminal communication \cite{Rohlich-Popescu, HR19}. While the no-signaling principle alone is not sufficient to pick out the set of quantum non-local correlations, local quantum mechanics and no-signaling together exactly reconstruct the set of quantum correlations in the two-party Bell scenario \cite{BBBEW10, Unified}.

Non-locality considers a black box scenario where the parties perform measurements chosen as random classical inputs and obtain corresponding classical outcomes. A related notion of quantum steering \cite{S36,WJD07} considers a more refined scenario where one or more of the parties are considered to have full control of the quantum systems in their laboratory, so that the quantum states and measurements of their subsystems are fully characterized. In this paradigm which has recently gained interest from both fundamental and applied perspectives \cite{CS17, UCNG20}, instead of just the measurement statistics one considers \textit{quantum assemblages} consisting of subnormalized states describing conditional states of the characterized subsystem conditioned upon measurements and outcomes of separated untrusted parties. Similarly, as an analog of the no-signaling boxes, one considers here the \textit {no-signaling assemblages} that are only constrained by the no-signaling conditions \cite{SBCSV15}. In particular, the generalization of the bipartite steering scenario to the multipartite case (with several untrusted systems) \cite{CS15} has lead to the possibility of post-quantum steering, i.e., of no-signaling assemblages that do not admit quantum realization \cite{SBCSV15, SAPHS18, HS18, SHSA20}.

While the study of no-signaling correlations has lead to a deep information-theoretic understanding of quantum non-locality \cite{AQ, IC}, several fundamental no-go theorems have also been discovered for proving device-independent security against adversaries constrained only by the no-signaling principle \cite{RTHHPRL, RA12}. Thus, from both a fundamental and applied perspective, it is of great interest to study the set of Bell correlations in scenarios where some subsystems obey quantum mechanics locally while the remaining subsystems are constrained only by the no-signaling principle. In this paper, we study this set of hybrid no-signaling-quantum correlations showing several interesting properties of this set that lend themselves naturally to (one-sided)-device-independent applications. 

The paper is organized as follows. We first introduce the set of hybrid no-signaling-quantum $\textbf{HNSQ}$ correlations defining these as the set of correlations obtainable by performing quantum measurements on the characterized subsystems of general no-signaling assemblages. After showing how $\textbf{HNSQ}$ fits in between the set of quantum non-local correlations $\textbf{Q}$ and the set of general no-signaling $\textbf{NS}$ correlations, we study interesting properties of the set. In particular, we show that in contrast to the usual $\textbf{NS}$ correlations, there exist extremal points of $\textbf{HNSQ}$ that admit a quantum realization. Furthermore, we show that some of these extreme points serve as self-testing certificates for boxes. Crucially, these results hold in a three-party setting where quantum assemblages
have been shown to be a strict subset of the set of no-signaling assemblages.
These surprising results, in view of the unrealisability of super-quantum boxes such as the PR box in non-locality \cite{RMPBellnonlocality, RMPBuhrman}, should have interesting consequences both in quantum foundations and in the development of one-sided-device-independent cryptography secure against super-quantum adversaries.

As a tool for optimization over the set $\textbf{HNSQ}$, we introduce an outer hierarchy of semi-definite programming approximations in Appendix G \cite{Supp} to the set following an approach put forward by Doherty-Parillo-Spedalieri in \cite{DPS, DPS2}. We recapitulate their results in Appendix E \cite{Supp}, and in Appendix F \cite{Supp} we provide a generalization of the DPS method for optimization over multi-partite state that are normalized but allow for negative eigenvalues. We use this tool to perform an extensive numerical analysis of the maximal violation of the facet Bell inequalities in the three-party binary input-output scenario \cite{Sliwa} and study the corresponding self-testing properties. \newline

\noindent\textit{Defining Q-NS correlations from No-Signaling Assemblages.-} 
We begin with a scenario, usually associated with quantum steering \cite{S36,WJD07}, in which a joint system consisting of a local (trusted) $d$ dimensional quantum subsystem $B$ together with $n$ distant untrusted subsystems $A_i$, is described by some theory (possibly beyond quantum mechanics) fulfilling no-signaling constrains. Assume that on each untrusted subsystem $A_i$ one can performed random measurements labeled by settings $x_i\in \left\{0,\ldots m_i-1\right\}$, and related outcomes $a_i\in \left\{0,\ldots k_i\right\}$, of cardinality $m_i$ and $k_i$, respectively, for $i=1,\ldots, n$. We will write $\textbf{a}_n=(a_1,\ldots, a_n)$ and $\textbf{x}_n=(x_1,\ldots, x_n)$ for strings of $n$ chosen settings and outcomes. In what follows we will introduce the simplified notation $k_i=k, m_i=m$ for all $i$, as it will not affect the generality of future arguments. When the described measurements are performed on the untrusted subsystems, the conditional state of the trusted subsystem $B$ is fully characterized by the notion of a no-signaling assemblage.

Specifically, a \textit{no-signaling assemblage} $\Sigma=\left\{\sigma_{\textbf{a}_n|\textbf{x}_n}\right\}$ is a collection of subnormalized states $\sigma_{\textbf{a}_n|\textbf{x}_n}$ acting on $d$-dimensional Hilbert space, for which there is a state $\sigma_B$ such that
\begin{equation}\label{def11}
\forall_{x_1,\ldots, x_n} \sum_{a_1,\ldots, a_n} \sigma_{\textbf{a}_n|\textbf{x}_n}=\sigma_B,
\end{equation}and for any possible subset of indexes $I=\left\{i_1,\ldots, i_s\right\}\subset \left\{1,\ldots, n\right\}$ with $1\leq s<n$, there exists a positive operator $\sigma_{a_{i_1}\ldots a_{i_s}|x_{i_1}\ldots x_{i_s}}$ such that the following no-signaling constraints hold:
\begin{equation}\label{def12}
\forall_{x_1,\ldots, x_n} \sum_{a_j,j\notin I} \sigma_{\textbf{a}_n|\textbf{x}_n}=\sigma_{a_{i_1}\ldots a_{i_s}|x_{i_1}\ldots x_{i_s}}.
\end{equation}
The convex set of all no-signaling assemblages (for a fixed dimension $d$ and fixed scenario $(n,m,k)$) will be denoted by $\textbf{nsA}(n,m,k,d)$. Note that no-signaling assemblage may have a quantum realization 
\begin{equation}\label{assemblage}
\sigma_{\textbf{a}_n|\textbf{x}_n}=\mathrm{Tr}_{A_1,\ldots, A_n}(M^{(1)}_{a_1|x_1}\otimes\ldots \otimes M^{(n)}_{a_n|x_n}\otimes \mathds{1}\rho),
\end{equation}
where $\rho\in \bigotimes_{i=1}^n B(H_{A_i})\otimes B(H_B)$ is some quantum state and all $M^{(i)}_{a_i|x_i}$ are elements of POVMs corresponding to appropriate measurement outcomes and settings. It is well-known \cite{G89,HJW93} that any no-signaling assemblage for $n=1$ admits a quantum realization of the form (\ref{assemblage}). In contrast, the $n>1$ scenario provides a room for post-quantum steering, i.e. quantum assemblages form a nontrivial convex subset $\textbf{qA}(n,m,k,d)$ of $\textbf{nsA}(n,m,k,d)$ \cite{SBCSV15}.\newline

On the other hand, we show in Theorem \ref{uni} that any no-signaling assemblage admits a realization analogous to the one given by (\ref{assemblage}), as long as one relaxes the positivity constraint for $\rho$.

\begin{theorem}\label{uni}Collection of positive operators $\Sigma=\left\{\sigma_{\textbf{a}_n|\textbf{x}_n}\right\}$ acting on $d_B$ dimensional space defines a no-signaling assemblage if and only if there exists a Hermitian operator $W\in \left[ \bigotimes_{i=1}^{n}B(\mathbb{C}^d) \right]\otimes B(\mathbb{C}^{d_B})$ of a unit trace, and POVMs elements $M^{(i)}_{a_i|x_i}\in B(\mathbb{C}^d)$, such that:
\begin{equation}
	\sigma_{\textbf{a}_n|\textbf{x}_n}=\mathrm{Tr}_{A_1,\ldots, A_n}(M^{(1)}_{a_1|x_1}\otimes\ldots \otimes M^{(n)}_{a_n|x_n}\otimes \mathds{1}W).
\end{equation}
Moreover, for any $\Sigma\in \textbf{nsA}(n,m,k,d_B)$ one can always chose a representation with the same choice of POVMs elements acting on $d$-dimensional space with $d=\max (m,k)$.
\end{theorem}
\begin{proof}See Appendix A in \cite{Supp}.
\end{proof}

Note that a similar observation has been made in \cite{SAPHS18}. However, that construction (based on the so-called pseudo-LHS model \cite{SAPHS18}) requires that the local dimensions $d$ fulfill $d=mk+1$, and as such it is less useful from the operational perceptive. 

Moreover, the presented theorem provides that by fixing the choice of POVMs (see construction of $W$, Appendix B in \cite{Supp}), one can restrict the choice of a Hermitian operator by imposing an upper bound of its operator norm, which ensures another commodity for numerical analysis.\newline

\noindent\textit{Hybrid no-signaling-quantum correlations.-} Consider now a Bell-type scenario with $n$ separated (untrusted) systems as above. The scenario is described by measurements $x_i\in \left\{0,\ldots m_i-1\right\}$ and outcomes $a_i\in \left\{0,\ldots k_i\right\}$ for each $i=1,\ldots, n$. By $\textbf{NS}(n,m,k)$ (here we once more use a simplified notation with single $m$ and $k$) we denote the convex set of all no-signaling boxes $P=\left\{p(\textbf{a}_n|\textbf{x}_n)\right\}$ that may describe this experimental setup. No-signaling correlations which admit a quantum realization given by $p(\textbf{a}_n|\textbf{x}_n)=\mathrm{Tr}(M^{(1)}_{a_1|x_1}\otimes\ldots \otimes M^{(n)}_{a_n|x_n}\rho)$ form a convex subset denoted as $\textbf{Q}(n,m,k)$. Finally, the set of local boxes $\textbf{LOC}(n,m,k)$ is defined as a convex hull of deterministic correlations in $(n,m,k)$ scenario. It is well known, that $\textbf{LOC}(n,m,k)\subsetneq \textbf{Q}(n,m,k) \subsetneq \textbf{NS}(n,m,k)$.

Now, we observe that the possibility of post-quantum steering (existence of post-quantum no-signaling assemblages) discussed in the previous part enables us to consider a new set of hybrid quantum no-signaling Bell correlations $\textbf{HNSQ}(n_{ns}+n_q,m,k)$ obtained from no-signaling assemblages by performing quantum measurements on one trusted subsystem.

\begin{definition}\label{NSA} 
The set of hybrid quantum no-signaling correlations $\textbf{HNSQ}(n_{ns}+n_q,m,k)$ is defined as a set of all no-signaling boxes $P=\left\{p(\textbf{a}_{n_{ns}+n_q}|\textbf{x}_{n_{ns}+n_q})\right\}$ such that $p(\textbf{a}_{n_{ns}+n_q}|\textbf{x}_{n_{ns}+n_q})=\mathrm{Tr}(M^{(1)}_{a_{1}|x_{1}}\otimes\ldots \otimes M^{(n_q)}_{a_{n_q}|x_{n_q}}\sigma_{\textbf{a}_{ns}|\textbf{x}_{ns}})$ for some $\Sigma=\left\{\sigma_{\textbf{a}_{ns}|\textbf{x}_{ns}}\right\}\in \textbf{nsA}(n_{ns},m,k,\prod_i^{n_q}d_i)$ where $M^{(i)}_{a_{i}|x_{i}}$ are elements of some POVMs acting on $d_i$-dimensional space respectively and $d_1,\ldots, d_{n_q}$ are arbitrary.
\end{definition} 

For the sake of simplicity we will not consider the above definition in full generality. In the remaining part of the paper we restrict our attention to the case with only single trusted (and quantum) subsystem and related set of hybrid correlation $\textbf{HNSQ}(n+1,m,k)$.

We first observe that $\textbf{HNSQ}(n+1,m,k)$ is convex. Indeed, let $P_1,P_2\in \textbf{HNSQ}(n+1,m,k)$ and choose any $q\in (0,1)$. By Definition \ref{NSA} we may put $p^{(i)}(\textbf{a}_{n+1}|\textbf{x}_{n+1})=\mathrm{Tr}(M^{(i)}_{a_{n+1}|x_{n+1}}\sigma^{(i)}_{\textbf{a}_{n}|\textbf{x}_{n}})$ for $i=1,2$, where $\Sigma_i=\left\{\sigma^{(i)}_{\textbf{a}_{n}|\textbf{x}_{n}}\right\}\in \textbf{nsA}(n,m,k,d_i)$. Consider $P=qP_1+(1-q)P_2$. Observe that 
\begin{equation}\nonumber
p(\textbf{a}_{n+1}|\textbf{x}_{n+1})=\mathrm{Tr}(N_{a_{n+1}|x_{n+1}}\sigma_{\textbf{a}_{n}|\textbf{x}_{n}})
\end{equation}where $N_{a_{n+1}|x_{n+1}}=M^{(1)}_{a_{n+1}|x_{n+1}}\oplus M^{(2)}_{a_{n+1}|x_{n+1}}\in B(\mathbb{C}^{d_1}\oplus \mathbb{C}^{d_2})$ are elements of some new POVMs and $\sigma_{\textbf{a}_{n}|\textbf{x}_{n}}=q\sigma^{(1)}_{\textbf{a}_{n}|\textbf{x}_{n}}\oplus (1-q)\sigma^{(2)}_{\textbf{a}_{n}|\textbf{x}_{n}}\in B(\mathbb{C}^{d_1}\oplus \mathbb{C}^{d_2})$. Obviously $\Sigma=\left\{\sigma_{\textbf{a}_{n}|\textbf{x}_{n}}\right\}\in \textbf{nsA}(n,m,k,d_1+d_2)$, so $P\in \textbf{HNSQ}(n+1,m,k)$.

From the characterization of $\textbf{nsA}(n,m,k,d)$ and $\textbf{qA}(n,m,k,d)$ we immediately have $\textbf{Q}(n+1,m,k)\subset \textbf{HNSQ}(n+1,m,k) \subset \textbf{NS}(n+1,m,k)$. The following theorems ensure that in fact all the above inclusions are strict, so $\textbf{HNSQ}(n+1,m,k)$ can be seen as a nontrivial intermediate theory of no-signaling boxes giving rise to post-quantum resources (see also Fig. \ref{fig_1}).

\begin{theorem}$\textbf{HNSQ}(n+1,m,k)\subsetneq  \textbf{NS}(n+1,m,k)$ for all $n,m,k$. Moreover, $\textbf{Q}(n+1,m,k)\subsetneq \textbf{HNSQ}(n+1,m,k)  $ for all $(n,m,k)$ with $n\geq 2$.
\end{theorem}
\begin{proof}See Appendix C in \cite{Supp}.
\end{proof}

\begin{figure}[H]
\includegraphics[width=0.45\textwidth]{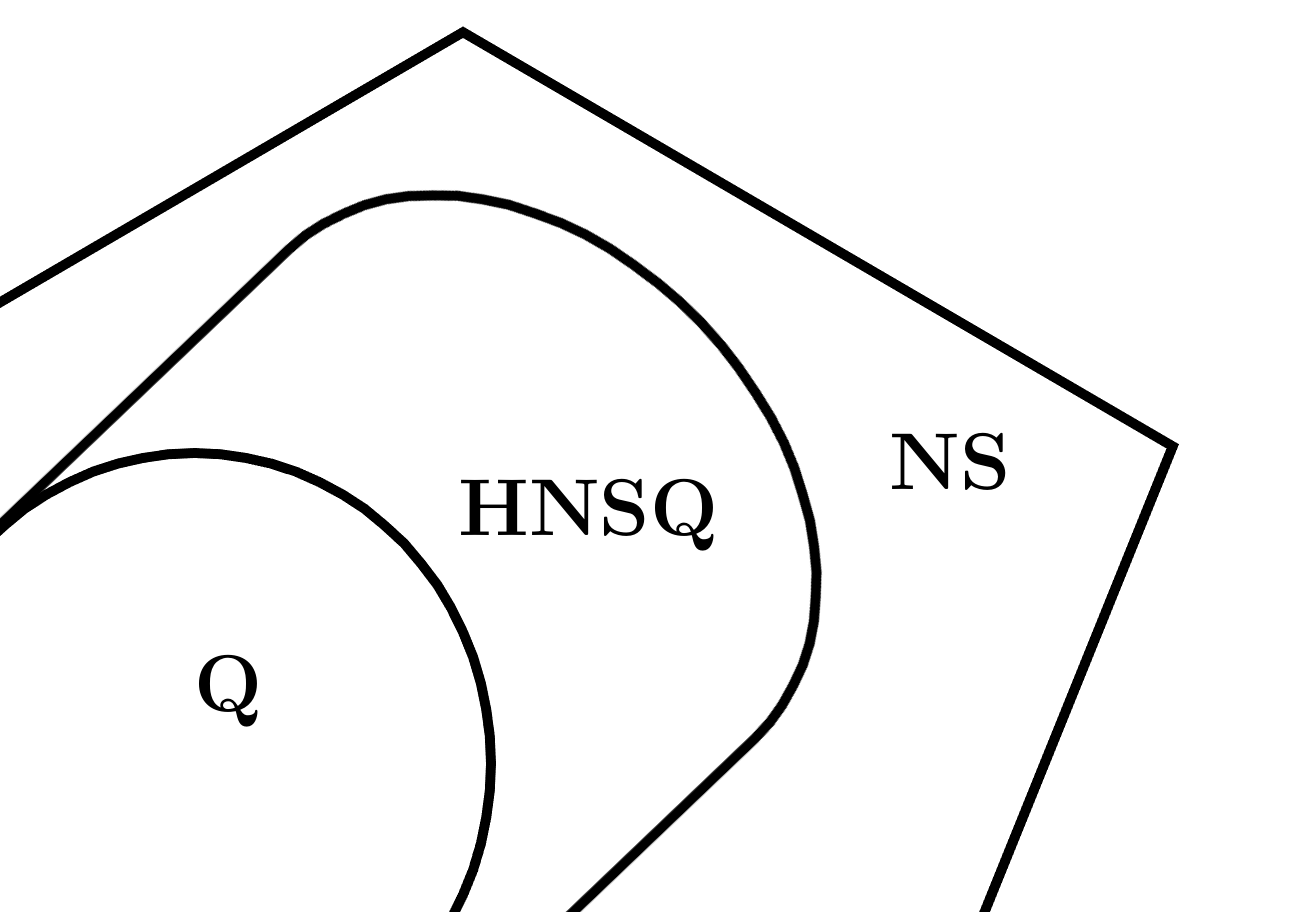}
\caption{\label{fig_1}Schematic relationship between $\textbf{Q},\textbf{HNSQ}$ and $\textbf{NS}$.}
\end{figure}

\begin{remark}\label{rem}Note that due to Jordan's lemma \cite{Jordan}, we have $\textbf{HNSQ}(n+1,2,2)=\mathrm{conv}(\textbf{NSA}_2(n+1,2,2))$, where $\textbf{NSA}_2(n+1,2,2)$ denotes a set of correlation obtained by projective measurements on assemblages acting only on a $2$-dimensional Hilbert space.
\end{remark}

\noindent\textit{Quantum realization of extreme points.-} It has been established in \cite{RTHHPRL} that it is not possible to have a quantum realization of a non-local (non-classical) extreme point in the set of no-signaling correlations (regardless of the numbers of parties, settings, and outcomes). Recently, the analogous question has been found to have a positive answer in the simplest nontrivial steering scenario with no-signaling assemblages \cite{RBRH}, that is, there are examples of quantum assemblages which are extremal in $\textbf{nsA}(2,2,2,d)$ and non-classical (where classicality is defined as admitting a local hidden state model \cite{SAPHS18}).

Considering the case of the intermediate set of hybrid no-signaling-quantum correlations $\textbf{HNSQ}(n+1,m,k)$, it is natural to pose the analogous question - is it possible to obtain a quantum realization of a non-local extreme point in $\textbf{HNSQ}(n+1,m,k)$? The following theorem shows a surprising positive solution for this problem.

\begin{theorem}\label{thm_ext}Let $P\in \textbf{Q}(3,m,k)$ be given by
\begin{equation}\label{form_of}
p(abc|xyz)=\mathrm{Tr}_{ABC}(P_{a|x}\otimes Q_{b|y}\otimes R_{c|z}|0_A\rangle \langle 0_A|\otimes|\psi_{BC}\rangle \langle \psi_{BC}|)
\end{equation}with projections $P_{f(x)|x}=|0\rangle \langle 0|$ defined by some fixed function $f(x)\in \left\{0,\ldots, k-1\right\}$ and where projections $Q_{b|y},R_{c|z}$ together with a state $|\psi_{BC}\rangle \langle \psi_{BC}|$ give a rise to the maximal quantum violation of a given bipartite Bell inequality $\mathcal{B}$ that is self-testable in the set of quantum correlations. Then $P$ is extremal in $\textbf{HNSQ}(2+1,m,k)$
\end{theorem}
\begin{proof}For detailed proof see Appendix D in \cite{Supp}.
\end{proof}

For a given Bell inequality $\mathcal{B}$ let $\Boxes_{\textbf{T}}(\mathcal{B}, B)$ denote the set of all boxes of the theory $\textbf{T}$ (understood as a convex set of all allowed correlations) for which the value of the Bell operator $\mathcal{B}$ attains at least the value $B$. Define Tsirelson's bound as $\Tsn_{\textbf{T}}(\mathcal{B})\equiv\max_{P \in \textbf{T}} \mathcal{B}(P)$. Those terms allows as to introduce a new form of self-testing, \textit{viz.} self-testing for boxes:
\begin{definition}
	A Bell operator $\mathcal{B}$ is self-tested for boxes within the theory $\textbf{T}$ if and only if the set $\Boxes_{\textbf{T}}(\mathcal{B}, \Tsn_{\textbf{T}}(\mathcal{B}))$ has cardinality $1$, i.e. when the box attaining the Tsirelson's bound is unique.
\end{definition}
In our considerations we will use the following lemma for tripartite boxes with two settings and outcomes (generalization for more parties, settings and outcomes is straightforward):
\begin{lem}
	For given theory $\textbf{T}$, Bell operator $\mathcal{B}$ and $B \in \mathbb{R}$ let:
	\begin{equation}
		\Flex_{\textbf{T}}(\mathcal{B}, B) \equiv \frac{1}{8} \sum_{a,b,c,x,y,z} \max_{p \in \Boxes_{\textbf{T}}(\mathcal{B}, B)}p(a,b,c|x,y,z).
	\end{equation}
	Then, it holds that $\mathcal{B}$ is self-tested for boxes in $\textbf{T}$ if and only if $\Flex_{\textbf{T}}(\mathcal{B}, \Tsn_{\textbf{T}}(\mathcal{B}))=1$.
\end{lem}
\begin{proof}
	Assume that $\mathcal{B}$ is not self-tested for boxes in $\textbf{T}$, so there exist at least two different boxes $p_1, p_2 \in \Boxes_{\textbf{T}}(\mathcal{B}, B)$. Let $(a_1,b_1,c_1|x_1,y_1,z_1)$ be a tuple of settings and outcomes for that $p_1(a_1,b_1,c_1|x_1,y_1,z_1) > p_2(a_1,b_1,c_1|x_1,y_1,z_1)$ (from normalization of probabilities it follows that for two different probability distributions such tuple exist). Then
	\begin{equation}
		\begin{aligned}
			\Flex_{\textbf{T}}(\mathcal{B}, B) \geq \frac{1}{8} \left[ \sum_{\substack{(a,b,c,x,y,z) \neq \\ (a_1,b_1,c_1,x_1,y_1,z_1)}} p_2(a,b,c|x,y,z) \right] + \\
			p_1(a_1,b_1,c_1|x_1,y_1,z_1) > \frac{1}{8} \sum_{a,b,c,x,y,z} p_2(a,b,c|x,y,z) = 1.
		\end{aligned}
	\end{equation}
\end{proof}


Let $\mathcal{B}$ be any bipartite Bell inequality for subsystem BC that is self-testable within the quantum theory (like in the above theorem). For any $P\in \textbf{HNSQ}(2+1,m,k)$ define a new tripartite inequality given by
\begin{equation}
\mathcal{I}_{\alpha}(P)=\mathcal{B}(P)-\sum_{x,y,z=0}^{m-1}\sum_{b,c=0}^{k-1}\sum_{a\neq f(x)}\alpha_{abc|xyz}p(abc|xyz)
\end{equation}where $\alpha_{abc|xyz}>0$ for any $a,b,c,x,y,z$ and $\mathcal{B}(P)$ stands for value of $\mathcal{B}$ computed on marginal conditional probabilities for subsystem BC coming from tripartite box $P$. One can show that $\Tsn_{\textbf{HNSQ}(2+1,m,k)}(\mathcal{I}_{\alpha})=\Tsn_{\textbf{Q}(3,m,k)}(\mathcal{I}_{\alpha})=\Tsn_{\textbf{Q}(2,m,k)}(\mathcal{B})$ and as a consequence we obtain the following observation. 

\begin{observation}Inequality $\mathcal{I}_{\alpha}$ serves as a self-testing certificate for boxes of the extremal box within $\textbf{HNSQ}(2+1,m,k)$ with a quantum realization (\ref{form_of}) given with respect to the bipartite Bell inequity $\mathcal{B}$ defining $\mathcal{I}_{\alpha}$.
\end{observation}

It easy to see that this line of reasoning can be generalized to the case of $\textbf{HNSQ}(n+1,m,k)$ with an arbitrary number of untrusted no-signaling parties $n$ when a considered box is a product of some deterministic $(n-1)$-partite box and some bipartite box that is self-testable within the quantum theory.

Note that these results open new potential one-sided cryptographic applications. While quantum correlations cannot provide extremality-based full security within a fully no-signaling theory, they can still do so against attacks within the restricted super-quantum $\textbf{HNSQ}$ set. \newline

\noindent\textit{Numerical results.-} We have numerically calculated maximal values for the set of Bell operators provided by \'{S}liwa \cite{Sliwa} and related values of $\Flex_{\textbf{T}}$ parameter allowed within sets $\textbf{Q}, \textbf{HNSQ}$ and $\textbf{NS}$.
Upper bounds for an intermediate set $\textbf{HNSQ}$ has been obtained with a proposed method from Appendix G in \cite{Supp} using DPS \cite{DPS,DPS2} with symmetric extension involving two copies of the measurement of Charlie and PPT-constraints imposed, see the Appendix H in  \cite{Supp} for presentation of proposed method for binary measurements and outcomes. (Note that alternative approach to upper approximation of intermediate set $\textbf{HNSQ}$ in the general case can be also preformed based on the notion of assemblage moment matrices \cite{Liang1,Liang2}.) The lower bound has been obtained using the see-saw method (for an explanation, see \cite{PV10}) and confirms the upper bound in all cases.

Our results confirm that the $\textbf{HNSQ}$ theory is strictly between $\textbf{Q}$ and $\textbf{NS}$. For example for \'{S}liwa's Bell operators 5, 8, 9, 12, 14, 15, 16, 42 the $\textbf{HNSQ}$ and $\textbf{NS}$ Tsirelson's bounds are the same and strictly larger than $\textbf{Q}$. This shows a striking property that $\textbf{HNSQ}$, contrary to $\textbf{Q}$, contains extremal no-signaling non-local boxes.

On the other hand, for Bell operators 6, 20, 45, both the quantum and $\textbf{HNSQ}$ Tsirelson's bounds are the same (strictly smaller than $\textbf{NS}$), meaning that in laboratory it is possible to realize experiments that are secure against $\textbf{HNSQ}$ adversary. We observe that for operator 20 and 45 the value of $\Flex$ in $\textbf{HNSQ}$ is equal to $1$ meaning that we provide a self-testing for boxes of a quantumly-realizable boxes within super-quantum theory.

Operators 7, 10, 18, 19, 21, 24, 25, 26, 27, 32, 36, 39 and 40 have different Tsirelson's bound for each theory.\newline

\textit{Conclusions.-} In this paper, we have introduced and studied a set of hybrid no-signaling-quantum Bell correlations obtained by performing quantum measurements on trusted parts of no-signaling assemblages. We introduced a tool for optimization over the set, namely a hierarchy of semi-definite programming based outer approximations. As a central interesting property of this set of super-quantum correlations that makes it more appealing than the set of no-signaling correlations, we have proved that there exist extremal points (boxes) in this set that admit quantum realization, and that furthermore these realizations are self-testing for boxes in some cases. A central open question is to formulate (one-sided)-device-independent protocols against super-quantum adversaries that are restricted to prepare boxes within the hybrid no-signaling-quantum correlation set. In particular, it would be interesting to investigate if these new correlations allow to circumvent some of the no-go theorems against super-polynomial privacy amplification that were proven in the scenario of general no-signaling adversaries \cite{RA12}.

\begin{acknowledgments}
\textit{Acknowledgments}- M.B., P. M. and P.H. acknowledge support by the Foundation for Polish Science (IRAP project, ICTQT, contract no. 2018/MAB/5, co-financed by EU within Smart Growth Operational Programme). M. B. acknowledges partial support from by DFG (Germany) and NCN (Poland) within the joint funding initiative Beethoven2 (Grant No. 2016/23/G/ST2/04273). R.R. acknowledges support from the Start-up Fund 'Device-Independent Quantum Communication Networks' from The University of Hong Kong, the Seed Fund 'Security of Relativistic Quantum Cryptography' and the Early Career Scheme (ECS) grant 'Device-Independent Random Number Generation and Quantum Key Distribution with Weak Random Seeds'.
\end{acknowledgments}


\appendix

\section{Proof of the Theorem 1}\label{appA}

If each positive operator $\sigma_{\textbf{a}_n|\textbf{x}_n}$ can be realized by the aforementioned presentation $\sigma_{\textbf{a}_n|\textbf{x}_n}=\mathrm{Tr}_{A_1,\ldots, A_n}(M^{(1)}_{a_1|x_1}\otimes\ldots \otimes M^{(n)}_{a_n|x_n}\otimes \openone W)$ with some Hermitian operator $W$ and some local measurements $M^{(i)}_{a_i|x_i}$, it is easy to see that a collection $\Sigma=\left\{\sigma_{\textbf{a}_n|\textbf{x}_n}\right\}$ form a no-signaling assemblage (see conditions (1) and (2) in the main paper). This explains one implication.

The other implication follows from generalization of the construction presented in \cite{Unified}. Indeed, let us consider a first untrusted subsystem $A_1$ described by some Hilbert space $\mathbb{C}^d$. For any $a_1=0,\ldots, k-2, x_1=0,\ldots, m-1$ we define normalized vectors $|\phi_{a_1|x_1}\rangle\in \mathbb{C}^d$ in such a way that
\begin{equation}
	\begin{aligned}
		&\left\{N_i\right\}_i = \left\{\openone\right\} \cup \\
		&\left\{|\phi_{a_1|x_1}\rangle \langle \phi_{a_1|x_1}|:a_1=0,\ldots, k-2, x_1=0,\ldots, m-1\right\}
	\end{aligned}
\end{equation}
is a set of linearly independent operators (it is always possible for $d$ large enough, according to \cite{Unified}, one may put $d=\max (m,k)$). Taking appropriate $z_{a_1|x_1}>0$ we define $M^{(1)}_{k-1|x_1}=\openone-\sum_{a_1=0}^{k-2}z_{a_1|x_1}|\phi_{a_1|x_1}\rangle \langle \phi_{a_1|x_1}|\geq 0$ and put $M^{(1)}_{a_1|x_1}=z_{a_1|x_1}|\phi_{a_1|x_1}\rangle \langle \phi_{a_1|x_1}|$ for remaining $a_1=0,\ldots, k-2$. For
\begin{equation}
	\left\{N_i\right\}_i=\left\{\openone\right\}\cup \left\{M^{(1)}_{a_1|x_1}:a_1=0,\ldots, k-2, x_1=0,\ldots, m-1\right\}
\end{equation}
we introduce a dual basis (with respect to the Hilbert-Schmidt inner product) of Hermitian operators $\left\{\tilde{N}_j\right\}_j=\left\{\tilde{\openone}\right\}\cup \left\{\tilde{M}^{(1)}_{a_1|x_1}:a_1=0,\ldots, k-2, x_1=0,\ldots, m-1\right\}$, i.e. $\mathrm{Tr}(N_i\tilde{N}_j)=\delta_{ij}$. We repeat the same construction for each untrusted subsystem $A_i$ (for $i=1,\ldots, n$) described by the same Hilbert space $\mathbb{C}^d$. As a result we obtain the same set of POVMs elements acting on each untrusted subsystem, i.e. if $a_i=a_j,x_i=x_j$ then $M^{(i)}_{a_i|x_i}=M^{(j)}_{a_j|x_j}$ for all $i,j=1,\ldots, n$.

Let $\mathcal{J}_n$ denotes the family of all subsets of the set $\left\{1,\ldots, n\right\}$. We introduce the following Hermitian operator
\begin{widetext}
\begin{equation}\label{gen_form_of_W}
W=\sum_{ I\in \mathcal{J}_n, I=\left\{i_1,\ldots, i_l\right\}} \sum_{a_{i_1}=0}^{k-2}\ldots \sum_{a_{i_l}=0}^{k-2} \sum_{x_{i_1}=0}^{m-1}\ldots \sum_{x_{i_l}=0}^{m-1} \tilde{M}_{a_1|x_1}^{(1),I}\otimes \ldots \otimes \tilde{M}_{a_n|x_n}^{(n),I} \otimes \sigma_{I} 
\end{equation}
\end{widetext}where
\begin{equation}\label{sigma_I}
\sigma_I=\begin{cases}
 \sigma_{a_{i_1}\ldots a_{i_l}|x_{i_1}\ldots x_{i_l}}& \mbox{if $I\neq \emptyset$} \\
\sigma_B &\mbox{if $I=\emptyset$}
\end{cases}
\end{equation}and
\begin{equation}
\tilde{M}_{a_i|x_i}^{(i),I}=\begin{cases}
 \tilde{M}_{a_i|x_i}^{(i)}& \mbox{if $i \in I$} \\
\tilde{\openone} &\mbox{if $i \notin I$}
\end{cases}.
\end{equation}For example with $n=2$, formula $\ref{gen_form_of_W}$ gives explicit expression
\begin{align}
W=&\sum_{a_1,a_2=0}^{k-2}\sum_{x_1,x_2=0}^{m-1}\tilde{M}^{(1)}_{a_1|x_1}\otimes\tilde{M}^{(2)}_{a_2|x_2}\otimes \sigma_{a_1a_2|x_1x_2}\nonumber \\
&+\sum_{a_1}^{k-2}\sum_{x_1}^{m-1}\tilde{M}^{(1)}_{a_1|x_1}\otimes\tilde{\openone}\otimes \sigma_{a_1|x_1}\nonumber\\
&+\sum_{x_2=0}^{k-2}\sum_{a_2=0}^{m-1}\tilde{\openone}\otimes\tilde{M}^{(2)}_{a_2|x_2}\otimes \sigma_{a_2|x_2}+\tilde{\openone}\otimes \tilde{\openone}\otimes \sigma_B.
\end{align}

Straightforward calculations based on properties of $M_{a_i|x_i}^{(i)}$ and $\tilde{M}_{a_i|x_i}^{(i)}$ show that $\mathrm{Tr}(W)=1$ and $\sigma_{\mathbf{a}_n|\mathbf{x}_n}=\mathrm{Tr}_{A_1,\ldots A_n}(M^{(1)}_{a_1|x_1} \otimes\ldots \otimes M^{(n)}_{a_n|x_n}\otimes \openone W)$ for all possible settings $x_1,\ldots x_n$ and all choice of outcomes $a_i=0,\ldots, k-2$ while $i=1,\ldots ,n$. Note that linear relations ((1) and (2) in the main paper) fulfilled by elements (operators) $\sigma_{\textbf{a}_n|\textbf{x}_n}$ and $\sigma_B$ are the same as linear relations fulfilled respectively by non-negative numbers $p(\textbf{a}_n|\textbf{x}_n)$ and $1$ when $P=\left\{p(\textbf{a}_n|\textbf{x}_n)\right\}$ describes a no-signaling box. Because of that, according to the main result of \cite{Unified}, we see that in fact $\sigma_{\mathbf{a}_n|\mathbf{x}_n}=\mathrm{Tr}_{A_1,\ldots A_n}(M^{(1)}_{a_1|x_1} \otimes\ldots \otimes M^{(n)}_{a_n|x_n}\otimes \openone W)$ for all possible settings $x_1,\ldots x_n$ and outcomes $a_1,\ldots a_n$, as remaining operators $\sigma_{\textbf{a}_n|\textbf{x}_n}$ with $a_i=k-1$ for some $i$ can be express by linear combinations of $\sigma_I$ defined in (\ref{sigma_I}) (due to formulas (1) and (2) in the main paper). This concludes the proof of Theorem 1.

\section{Bound for operator norm of $W$ in a binary case}\label{appB}

In the discussion following Theorem 1 we stated, that by fixing all measurements $M^{(i)}_{a_i|x_i}$ for a given scenario, one can provide an upper bound on operator norms of all $W$ that are sufficient for realization of any no-signaling assemblage (in a given scenario $(n,m,k,d_B)$). Here we provide an explicit calculation of this bound in the case of binary measurements with binary outputs performed on each untrusted party.

According to construction of $W$ described in the proof of Theorem 1 (see the above Appendix \ref{appA}) it is enough to consider the same set of measurements on each untrusted subsystem. In the binary case we choose positive operators $M_{0|0},M_{0|1}$ in such a way that $\left\{\openone,M_{0|0},M_{0|1} \right\}$ is a linearly independent system. Dual system $\left\{\tilde{\openone},\tilde{M}_{0|0},\tilde{M}_{0|1} \right\}$ is given by 
$\tilde{\openone}=a_1\openone+a_2M_{0|0}+a_3M_{0|1}$, $\tilde{M}_{0|0}=b_1\openone+b_2M_{0|0}+b_3M_{0|1}$ and $\tilde{M}_{0|1}=c_1\openone+c_2M_{0|0}+c_3M_{0|1}$, where $a=(a_1,a_2,a_3)^T, b=(b_1,b_2,b_3)^T,c=(c_1,c_2,c_3)^T$ and $a=M^{-1}e_1, b=M^{-1}e_2, c=M^{-1}e_3$ with $\left\{e_i\right\}_i$ denoting the standard basis in $\mathbb{C}^3$ and 
\begin{equation}
M=\begin{pmatrix}
\begin{array}{ccc}
 \mathrm{Tr}(\openone) &  \mathrm{Tr}(M_{0|0}) & \mathrm{Tr}(M_{0|1}) \\  
 \mathrm{Tr}(M_{0|0})  & \mathrm{Tr}(M^2_{0|0}) & \mathrm{Tr}(M_{0|0}M_{0|1})  \\
 \mathrm{Tr}(M_{0|1})  & \mathrm{Tr}(M_{0|1}M_{0|0})  & \mathrm{Tr}(M^2_{0|1})  
\end{array}
\end{pmatrix}.
\end{equation}

We will restrict our attention to projective measurements. In other words without loss of generality (i.e. up to unitary transformation) we may put $M_{0|0}=|0\rangle \langle 0|, M_{0|1}=|\theta\rangle \langle \theta|$ with $|\theta\rangle=(\cos \theta, \sin \theta)^T$ where $\theta \in (0,\frac{\pi}{2})\cup (\frac{\pi}{2},\pi)$. From direct calculations based on the above description we obtain
\begin{equation}
\left\|\tilde{\openone}\right\|=\begin{cases}
\frac{1}{2\cos \theta}+\frac{1}{2} & \mbox{if $\theta \in (0,\frac{\pi}{2})$} \\
-\frac{1}{2\cos \theta}+\frac{1}{2} &\mbox{if $\theta \in (\frac{\pi}{2},\pi)$}
\end{cases}
\end{equation}and
\begin{equation}
\left\|\tilde{M}_{0|0}\right\|=\left\|\tilde{M}_{0|1}\right\|=|(2\sin \theta \cos \theta)^{-1}|.
\end{equation}Because of these particular forms, we may restrict our attention to the case $\theta \in (0,\frac{\pi}{2})$. With this choice and by the construction of $W$ presented in the proof of Theorem 1 in Appendix \ref{appA} (having in mind that for any no-signaling assemblage $\Sigma=\left\{\sigma_{\mathbf{a}_n|\mathbf{x}_n}\right\}$ we have $\left\|\sigma_{I}\right\| \leq 1$) we may use a subadditive property of the operator norm in order to show that (\ref{gen_form_of_W}) implies the following bound
\begin{widetext}
\begin{equation}\label{expw_thet}
\left\|W(\theta)\right\|\leq \sum_{k=0}^n \binom{n}{k}\left\|\tilde{\openone}\right\|^k \sum_{i=0}^{n-k}\binom{n-k}{i}\left\|\tilde{M}_{0|0}\right\|^{i}\left\|\tilde{M}_{0|1}\right\|^{n-k-i}=\sum_{k=0}^n \binom{n}{k}\left(\frac{1}{2\cos \theta}+\frac{1}{2}\right)^k\left(\frac{1}{\sin \theta\cos \theta}\right)^{n-k}=(f(\theta))^n
\end{equation}
\end{widetext}
with 
\begin{equation}
f(\theta)=\frac{1}{2}+\frac{1}{2\cos \theta}+\frac{1}{\sin \theta \cos \theta}.
\end{equation}

Expression (\ref{expw_thet}) attains minimal value for $\theta\in (0,\frac{\pi}{2})$ such that $f'(\theta)=0$ (regardless of $n$), i.e. when angle $\theta$ fulfills $\sin^3\theta +4\sin^2\theta-2=0$. Numerical calculations \cite{mathematica} shows that $\theta_{min}\approx 0.715$ and $f(\theta_{min})\approx 3.182$ (for comparison when $\theta_{+}=\frac{\pi}{4}$, i.e. $|\theta\rangle =|+\rangle$, one obtains $f(\theta_{+})=\frac{5+\sqrt{2}}{2}\approx 3.207$ - when $n$ is small this choice of angle $\theta_{+}$ leads to upper bound close to the optimal one).

\section{Proof of the Theorem 3}\label{appC}

Let us consider the fist statement of the theorem. The inclusion part is obvious, and as for $n=1$ we have $\textbf{HNSQ}(1+1,m,k)=\textbf{Q}(2,m,k)\subsetneq  \textbf{NS}(2,m,k)$ (see \cite{G89,HJW93} and \cite{RTHHPRL}), it is enough to examine $n\geq 2$. Put $n=2$ and consider any $P=\left\{p(abc|xyz)\right\}\in \textbf{HNSQ}(2+1,m,k)$. Obviously 
\begin{equation}
p(abc|xyz)=\mathrm{Tr}_{C}(M_{c|z}\sigma_{ab|xy})
\end{equation} for some POVMs elements $M_{c|z}$ and no-signaling assemblage defined by $\sigma_{ab|xy}$. Observe that in particular
\begin{equation}\label{obstruction}
p(bc|yz)=\mathrm{Tr}_{C}(M_{c|z}\sigma_{b|y})=\mathrm{Tr}_{BC}(N_{b|y}\otimes M_{c|z}|\psi_{BC}\rangle \langle \psi_{BC}|)
\end{equation}as elements $\sigma_{b|y}$ define a bipartite no-signaling assemblage, which always admits a quantum realization \cite{G89,HJW93}. Construct now a box $\tilde{P}=\left\{\tilde{p}(abc|xyz)\right\}\in \textbf{NS}(3,m,d)$ stating from $R=\left\{r(bc|yz)\right\}\in \textbf{NS}(2,m,k)$ such that $R\notin \textbf{HNSQ}(1+1,m,k)=\textbf{Q}(2,m,k)$ and putting
\begin{equation}\label{obstruction}
\tilde{p}(abc|xyz)=d(a|x)r(bc|yz)
\end{equation}where $d(a|x)$ define some fixed deterministic box. As (\ref{obstruction}) holds, we see that $\tilde{P}\notin \textbf{HNSQ}(2+1,m,k)$, since marginal probability $\tilde{p}(bc|yz)=r(bc|yz)$ does not have a quantum realization. One may repeat this procedure inductively, starting construction for step $n+1$ from the correlations box constructed in the previous step $n$.

Now let us consider the second statement of the theorem. The inclusion part is obvious. For $n\geq 2$ consider $P=\left\{p(a_1\ldots a_{n+1}|x_1\ldots x_{n+1})\right\}\in \textbf{HNSQ}(n+1,m,k)$ such that the marginal box $\tilde{P}=\left\{p(a_1\ldots a_{n}|x_1\ldots x_{n})\right\}\in \textbf{NS}(n,m,k)$ is non-local and extremal in $\textbf{NS}(n,m,k)$. According to the main results of \cite{RTHHPRL}, $\tilde{P}\notin \textbf{Q}(n,m,k)$ and we see that $P\notin \textbf{Q}(n+1,m,k)$ as well. This concludes the proof.

\section{Proof of the Theorem 5}\label{appD}
Observe that 
\begin{equation}\label{primal}
p(abc|xyz)=\delta_{a,f(x)}p(bc|yz)
\end{equation}where quantum correlations
\begin{equation}\label{primal1}
p(bc|yz)=\mathrm{Tr}_{BC}(Q_{b|y}\otimes R_{c|z}|\psi_{BC}\rangle \langle \psi_{BC}|)
\end{equation}obtain a quantum maximum for some fixed inequality $\mathcal{B}$ with desired properties. Consider a presentation of $P$ as a convex combination of elements from the set of hybrid no-signaling-quantum correlations, i.e.
\begin{equation}\label{conv}
p(abc|xyz)=\sum_i^n q_ip^{(i)}(abc|xyz)=\sum_i^n q_i \mathrm{Tr}_{C}(R^{(i)}_{c|z}\sigma^{(i)}_{ab|xy})
\end{equation}where $q_i>0, \sum_i^n q_i=1$, $\Sigma_i=\left\{\sigma^{(i)}_{ab|xy}\right\}$ are some no-signaling assemblages and $R^{(i)}_{c|z}$ stand for elements of some POVMs.
Note that due to (\ref{primal}) and (\ref{conv}) we obtain
\begin{equation}
p(ab|xy)=\delta_{a,f(x)}p(b|y)=\sum_i^n q_i \mathrm{Tr}_{C}(\sigma^{(i)}_{ab|xy})
\end{equation}and as a result
\begin{equation}\label{sig}
\sigma^{(i)}_{ab|xy}=\delta_{a,f(x)}\sigma_{b|y}^{(i)}.
\end{equation}On the other hand, we also have
\begin{equation}\label{primal2}
p(bc|yz)=\sum_i^n q_i p^{(i)}(bc|yz)=\sum_i^n q_i \mathrm{Tr}_{C}(R^{(i)}_{c|z}\sigma^{(i)}_{b|y}).
\end{equation}Since any $\tilde{\Sigma}_i=\left\{\sigma^{(i)}_{b|y}\right\}$ is a bipartite no-signaling assemblage, there is a quantum realization 
\begin{equation}
\sigma^{(i)}_{b|y}=\mathrm{Tr}_{B}(S^{(i)}_{b|y}\otimes \openone|\phi^{(i)}_{BC}\rangle \langle \phi^{(i)}_{BC}|)
\end{equation}where in general $S^{(i)}_{b|y}$ are elements of POVMs. If so, then by (\ref{primal2})
\begin{align}\label{form0}
p^{(i)}(bc|yz)&=\mathrm{Tr}_{BC}(S^{(i)}_{b|y}\otimes R^{(i)}_{c|z}|\phi^{(i)}_{BC}\rangle \langle \phi^{(i)}_{BC}|)\\ \nonumber
&=\mathrm{Tr}_{\tilde{B}\tilde{C}}(\tilde{S}^{(i)}_{b|y}\otimes \tilde{R}^{(i)}_{c|z}|\tilde{\phi}^{(i)}_{\tilde{B}\tilde{C}}\rangle \langle \tilde{\phi}^{(i)}_{\tilde{B}\tilde{C}}|)
\end{align}where states $|\tilde{\phi}^{(i)}_{\tilde{B}\tilde{C}}\rangle \langle \tilde{\phi}^{(i)}_{\tilde{B}\tilde{C}}|$ and PVMs elements $\tilde{S}^{(i)}_{b|y}, \tilde{R}^{(i)}_{c|z}$ come from Naimark's dilation of elements $S^{(i)}_{b|y}$ and $R^{(i)}_{c|z}$ respectively. In particular, $p^{(i)}(bc|yz)$ admits quantum realization (for any $i$).

Because of (\ref{primal1}) and (\ref{primal2}) each $p^{(i)}(bc|yz)$ maximizes $\mathcal{B}$. Since (\ref{form0}) holds, by the self-testing statement for projections $Q_{b|y}, R_{c|z}$ and state $|\psi_{BC}\rangle \langle \psi_{BC}|$, we obtain 
\begin{equation}
p^{(i)}(bc|yz)=p(bc|yz).
\end{equation}Finally, the above result together with (\ref{primal}) and (\ref{sig}) show that for any $i$
\begin{align}\nonumber
p^{(i)}(abc|xyz)&=\mathrm{Tr}_{C}(R^{(i)}_{c|z}\sigma^{(i)}_{ab|xy})=\delta_{a,f(x)}\mathrm{Tr}_{C}(R^{(i)}_{c|z}\sigma^{(i)}_{b|y})\\ \nonumber
&=\delta_{a,f(x)}p(bc|yz)=p(abc|xyz) \nonumber
\end{align}which ends the proof.

\section{Doherty-Parillo-Spedalieri hierarchy, the SWAP trick, and bilinear optimization}
\label{sec:DPSandSWAP}

In order to describe the outer approximations hierarchy, we first recapitulate the Doherty-Parillo-Spedalieri hierarchy (DPS)~\cite{DPS3,DPS,DPS2}. Let us consider a state $\rho$ on $\mathcal{H}_A \otimes \mathcal{H}_B \otimes \mathcal{H}_C \otimes \cdots$ that is separable, i.e. can be written as a convex combination of pure product states:
\begin{equation}
	\label{eq:Seperable}
	\rho = \sum_{i} \lambda_i \proj{\phi}_A \otimes \proj{\varphi}_B \otimes \proj{\chi}_C \otimes \cdots,
\end{equation}
where $\sum_{i} \lambda_i = 1$, $\lambda_i \geq 0$. Now, let $\tilde{\rho}$ be a state on $\mathcal{H}_{A^k} \otimes \mathcal{H}_{B^l} \otimes \mathcal{H}_{C^m} \otimes \cdots$, where $\mathcal{H}_{A^k}$ is a product of $k$ spaces $\mathcal{H}_A$, and similarly for $B$, $C$ etc. The state $\tilde{\rho}$ is an \textbf{extension} of $\rho$ if
\begin{equation}
	\rho = \Tr_{A^{k-1} B^{l-1} C^{m-1} \cdots} \left[ \tilde{\rho} \right],
\end{equation}
where the partial trace is over all but the first copy of each space. Let $\mathbb{S}_A$ be the set of all permutation operators between copies of the space $\mathcal{H}_A$, and similarly for $B$, $C$ etc. The state extension $\tilde{\rho}$ is symmetric if for all $P \in \mathbb{S}_A \otimes \mathbb{S}_B \otimes \mathbb{S}_C$ it holds that
\begin{equation}
	\tilde{\rho} = P \tilde{\rho} P.
\end{equation}
The state extension $\tilde{\rho}$ is PPT if $\tilde{\rho}$ remains positive after any partial transposition over subsystems. If $\rho$ is separable then for any $k, l, m, \cdots$, there exist $\tilde{\rho}$ that is a PPT symmetric extension of $\rho$. The core idea of the DPS hierarachy is to check whether for fixed $k, l, m, \cdots$ a PPT symmetric extension of $\rho$ exist; if not this means that $\rho$ is not separable.

Since the PPT symmetric extension constraints can be formulated as an SDP, the DPS method allows to optimize over a relaxation of the set of separable states on given spaces; the higher are $k, l, m, \cdots$, the relaxation is closer to the actual set of all separable states. Thus, from a PPT symmetric extension state $\tilde{\rho}$ we may construct a state $\varrho$ on $\mathcal{H}_A \otimes \mathcal{H}_B \otimes \mathcal{H}_C \otimes \cdots$ that is in some sense \textit{close} to be separable.

Further we denote by $\DPS\left(\mathcal{H}_A, \mathcal{H}_B, \mathcal{H}_C, \cdots\right)$ the set of all subnormalized states satisfying DPS criteria for some fixed level of that hierarchy. We then write (cf.~\eqref{eq:Seperable}):
\begin{equation}
	\rho = \sum_{i} \lambda_i \proj{\phi}_A \tilde{\otimes} \proj{\varphi}_B \tilde{\otimes} \proj{\chi}_C \tilde{\otimes} \cdots,
\end{equation}
 where $\tilde{\otimes}$ denotes the fact that the state is only \textit{close} to the product form, in the sense of DPS relaxation. In the case of higher dimensions, or more parties, one needs to include more subspaces to include optimization over more mesurements.

Another property that we will use is the relation~\cite{Werner89}:
\begin{equation}
	\label{eq:Werner89}
	\Tr \left[ \SWAP(A,A') \rho^{(1)}_{A} \otimes \rho^{(2)}_{A'} \right] = \Tr \left[ \rho^{(1)}_{A} \rho^{(2)}_{A} \right],
\end{equation}
where $\mathcal{H}_A$ and $\mathcal{H}_A'$ are isomorphic, and $\SWAP(A,A')$ is the SWAP operator between them. Below we denote space isomorphic to $A$ by $A'$ and $A''$, and similarly for $B$, $C$, etc.

In order to perform the optimization over boxes we need to optimize both over the state and measurements. Consider the following 
\begin{equation}
	\label{eq:WM}
	\begin{aligned}
		&\xi_{A'B'C'\cdots A''B''C''\cdots} \equiv \sum_i \lambda_i \kb{\phi_i}{\phi_i}_{A'B'C'\cdots} \otimes \\
		&\qquad \otimes \kb{\psi^A_i}{\psi^A_i}_{A''} \otimes \kb{\psi^B_i}{\psi^B_i}_{B''} \otimes \kb{\psi^C_i}{\psi^C_i}_{C''} \otimes \cdots,
	\end{aligned}
\end{equation}
where $\sum_i \lambda_i = 1$. From~\eqref{eq:Werner89} one can see that
\begin{equation}
	\begin{aligned}
		&\Tr_{A' A''} \left[ \SWAP(A,A') \xi_{A'B'C'\cdots A''B''C''\cdots} \right] = \\
		&\sum_i \lambda_i \Tr_{A'} \left[ \kb{\phi_i}{\phi_i}_{A'B'C'\cdots} \left( \kb{\psi^A_i}{\psi^A_i}_{A''} \otimes \openone_{B'C'\cdots} \right) \right],
	\end{aligned}
\end{equation}
thus the result is the post-measurement state after measurement on $\mathcal{H}_{A'}$. This construction easily generalizes if more measurements on each Hilbert space are possible, simply by adding subspaces ``storing'' their projectors. Thus all probabilities of measurements can be extracted as a linear function of $\xi_{A'B'C'C''}$, see~\cite{DPS}, eq.~(5).

In most cases, where only quantum systems are considered, one assumes in~\eqref{eq:WM} that $\lambda_i \geq 0$. Since SDP allows only to optimize for a target that is linear in variables, we combine DPS hierarchy and the SWAP trick to create an SDP variable approximating a joint state of a multipartite quantum state on systems $A'B'C'\cdots$ and a measurement represented on a system $A''B''C''\cdots$. On the other hand, if it doesn't hold that $\lambda_i \geq 0$ for all $i$, then the direct application of DPS is not possible.

\section{Bounded-norm operators optimization}
\label{sec:AppB}

For a given Hilbert space $\mathcal{H}$ let $\Dim\left(\mathcal{H}\right)$, $\Op\left(\mathcal{H}\right)$ and $\PSD\left(\mathcal{H}\right)$ be its dimension, the set of all bounded operators, and the set all bounded positive semidefinite operators acting on it, respectively. We also denote $\PSD\left(\mathcal{H}_1, \mathcal{H}_2, \cdots\right)$ as the set of operators of a form of product of semidefinite operators over those spaces, i.e.
\begin{equation}
	\PSD\left(\mathcal{H}_1, \mathcal{H}_2, \cdots\right) \equiv \PSD\left(\mathcal{H}_1\right) \otimes \PSD\left(\mathcal{H}_2\right) \otimes \cdots.
\end{equation}
For any $W \subseteq \Op\left(\mathcal{H}\right)$ let
\begin{equation}
	\Spec \left[ W \right] \equiv \inf \left\{ \Lambda \in \mathbb{R}_{+} : \forall_{w \in W} \left(-\Lambda \openone \preceq w \preceq \Lambda \openone \right) \right\},
\end{equation}
i.e. $\Spec \left[ W \right]$ is the boundary of the spectrum of all operators in $W$. Let $\Eig(w)$ for $w \in \Op\left(\mathcal{H}\right)$ be the set of all eigenvalues of $w$. Obviously,
\begin{equation}
	\Eig(w) \subseteq [-\Spec[\{w\}], \Spec[\{w\}]].
\end{equation}

Let
\begin{subequations}
	\begin{equation}
		\label{eq:L}
		\mathcal{L} : \Op\left( \mathcal{H}_X \otimes \mathcal{H}_{Y1} \otimes \mathcal{H}_{Y2} \otimes \cdots \right) \rightarrow \left[ \Op\left(\mathcal{H}_{Zi}\right) \right]_i,
	\end{equation}
	\begin{equation}
		\label{eq:T}
		\mathcal{T} : \Op\left(\mathcal{H}_Z\right) \rightarrow \mathbb{R},
	\end{equation}
\end{subequations}
be linear functionals, and let
\begin{equation}
	M_X \equiv \frac{1}{\Dim\left( \mathcal{H}_X \right)} \openone_X,
\end{equation}
be the maximally mixed normalized state operator on $\mathcal{H}_X$. In~\eqref{eq:L} the expression $\left[ \Op\left(\mathcal{H}_{Zi}\right) \right]_i$ denotes an indexed sequence of operators. We write $Z =  [Z_i]_i \succeq 0$ for $Z \in \left[ \Op\left(\mathcal{H}_{Zi}\right) \right]_i$ as a joint condition that $Z_i \succeq 0$ for all $i$.

Consider the following optimization problem for given $\left( W_X, \mathcal{L}, \mathcal{T} \right)$, where $W_X \subseteq \Op\left(\mathcal{H}_X\right)$ is some fixed set of operators on $\mathcal{H}_X$:
\begin{align}
	\label{eq:Opt1}
	\begin{split}
		\text{maximize } &\null \mathcal{T}(z), \\
		\text{over } &\null x \in W_X, \\
		&\null y \in \PSD\left( \mathcal{H}_{Y1}, \mathcal{H}_{Y2}, \cdots \right), \\
		\text{subject to } &\null \Tr[x] = 1, \\
		&\null \Tr[y] = 1, \\
		&\null z = \mathcal{L}(x \otimes y), \\
		&\null z \succeq 0.
	\end{split}
\end{align}

Let $\Lambda \equiv \Spec \left[ W_X \right]$, and $d \equiv \Dim\left( \mathcal{H}_X \right)$.
Now, consider the following optimization:
\begin{align}
	\label{eq:Opt2}
	\begin{split}
		\text{maximize } &\null \mathcal{T}(z), \\
		\text{over } &\null \xi \in \DPS\left( \mathcal{H}_X, \mathcal{H}_{Y1}, \mathcal{H}_{Y2}, \cdots \right), \\
		\text{subject to } &\null x = \Tr_Y(\xi), \\
		&\null x \preceq \frac{2 \Lambda}{1+\Lambda d} \openone_X, \\
		&\null y = \Tr_X(\xi), \\
		&\null z = \mathcal{L}\left[ (1+\Lambda d)\xi - \Lambda d M_X \otimes y \right], \\
		&\null z \succeq 0.
	\end{split}
\end{align}

\begin{lem}
	\label{lem:Opts}
	For any $\left( W_X, \mathcal{L}, \mathcal{T} \right)$ it holds that the solution of~\eqref{eq:Opt1} is upper bounded by the solution of~\eqref{eq:Opt2}.
\end{lem}
\begin{proof}
	We first show the following inclusion:
	\begin{equation}
		\begin{aligned}
			&\left\{ x \in W_X : \Tr[x] = 1 \right\} \subseteq \\
			&\left\{ x \in \Op\left(\mathcal{H}_X\right) : \Eig[x] \in [-\Lambda, \Lambda], \Tr[x] = 1 \right\} = \\
			&\left\{ x - \Lambda \openone_X : x \in \Op\left(\mathcal{H}_X\right), \Eig[x] \in [0, 2 \Lambda], \Tr[x] = 1+\Lambda d \right\} \\
			&\subseteq \tilde{W}_X \equiv \left\{ (1+\Lambda d)x - \Lambda d M_X : x \in \PSD\left(\mathcal{H}_X\right), \right. \\
			& \qquad \qquad \qquad \left. x \preceq 2 \Lambda/(1+\Lambda d) \openone_X, \Tr[x] \leq 1 \right\}.
		\end{aligned}
	\end{equation}
	The first inclusion follows from replacing the subset $W_X$ of operators on $\mathcal{H}_X$ bounded by $\Lambda$ with the set of all operators on $\mathcal{H}_X$ bounded by $\Lambda$.
	The equality follows from the shift of spectra by $\Lambda$ that is equivalent to adding the operator $\Lambda \openone_X$.
	The second inclusion follows replacement of $x$ with $\frac{1}{1+\Lambda d} x$, rewriting eigenvalue boundaries in matrix inequality form and relaxing the constraint $\Tr[X] = 1$.
	
	In this proof we will use the following abbreviation: $\PSD_Y \equiv \PSD\left( \mathcal{H}_{Y1}, \mathcal{H}_{Y2}, \cdots \right)$.
	
	It is easy to see that
	\begin{equation}
		\label{eq:xyBeforeRelax}
		\begin{aligned}
			&\left\{ x \otimes y : x \in W_X, y \in \PSD_Y, \Tr[x] = 1, \Tr[y] = 1 \right\} \subseteq \\
			&\left\{ x \otimes y : x \in \tilde{W}_X, y \in \PSD_Y, \Tr[x] = 1, \Tr[y] = 1 \right\} = \\
			&\left\{ \left[ (1+\Lambda d) x - \Lambda d M_X \right] \otimes y : x \in \PSD\left(\mathcal{H}_X\right), y \in \PSD_Y, \right. \\
			& \qquad \qquad \qquad \left. \Tr[x] = 1, \Tr[y] = 1, x \preceq 2 \Lambda/(1+\Lambda d) \openone_X \right\},
		\end{aligned}
	\end{equation}
	where the inclusion follows from $W_X \subseteq \tilde{W}_X$, and equality from direct writing the definition of $\tilde{W}_X$.
	
	The last expression of the set can also be rewritten as:
	\begin{equation}
		\label{eq:longSetXi}
		\begin{aligned}
			&\left\{ (1+\Lambda d) \xi - \Lambda d M_X \otimes \Tr_X[\xi] : \xi = x \otimes y, \right. \\
			&\qquad x \in \PSD\left(\mathcal{H}_X\right), y \in \PSD_Y, \\
			&\qquad \left. \Tr[x] = 1, \Tr[y] = 1, x \preceq 2 \Lambda/(1+\Lambda d) \openone_X \right\}.
		\end{aligned}
	\end{equation}
	
	Next, we use the fact that $\DPS\left( \mathcal{H}_X, \mathcal{H}_{Y1}, \mathcal{H}_{Y2}, \cdots \right)$ is a relaxation of the set of product states, i.e. it contains the set $\left\{ \xi \in \PSD\left(\mathcal{H}_X\right) \otimes \PSD_Y : \Tr[\xi] \leq 1 \right\}$. Thus we get that the set given in~\eqref{eq:longSetXi} is contained in
	\begin{equation}
		\label{eq:finalSetXi}
		\begin{aligned}
			&\left\{ (1+\Lambda d) \xi - \Lambda d M_X \otimes \Tr_X[\xi] : \xi \in \DPS\left(\mathcal{H}_X, \mathcal{H}_Y\right), \right. \\
			&\qquad \left. \Tr_Y[\xi] \preceq 2 \Lambda/(1+\Lambda d) \openone_X \right\}.
		\end{aligned}
	\end{equation}

	Direct comparison of the definition of the first set in~\eqref{eq:xyBeforeRelax} with the constraints in~\eqref{eq:Opt1}, and the definition of the set~\eqref{eq:finalSetXi} with the constraints in~\eqref{eq:Opt2} concludes the proof.
\end{proof}


\section{Hierarchy of outer semi-definite programming approximations of Hybrid no-signaling-quantum correlations}
\label{sec:Hierarchy}

In order to model the set of boxes in a hybrid scenario, we identify the Hilbert space of the operator $W$ from the Theorem 1 (in the main paper) with $\mathcal{H}_X$, spaces of operators $M^{(i)}_{a_i|x_i}\in B(\mathbb{C}^d)$ with $\mathcal{H}_{Y(a_i,x_i)}$, and the spaces of operators $M_{a_{n+1}|x_{n+1}}$ from the Definition 2 (in the main paper) with $\mathcal{H}_{Y(a_{n+1},x_{n+1})}$. We assume all the measurement operator to be projective with trace $1$.

Let us now define the operator $\mathcal{L}$ of~\eqref{eq:L} as the operator that transforms the state $W$ and the measurements $M^{(i)}_{a_i|x_i}\in B(\mathbb{C}^d)$, $i = 1,\cdots,n+1$, to the sequence
\begin{equation}
	\label{eq:SigmaMs}
	\left( \sigma_{\textbf{a}_n|\textbf{x}_n} \otimes M_{a_{n+1}|x_{n+1}} \right).
\end{equation}
This operator can be constructed using the SWAP trick~\eqref{eq:Werner89}.

Next, let $\mathcal{P}$ be a transformation from~\eqref{eq:SigmaMs} to probabilities $p(\textbf{a}_{n+1}|\textbf{x}_{n+1})=\mathrm{Tr}(M_{a_{n+1}|x_{n+1}}\sigma_{\textbf{a}_n|\textbf{x}_n})$ (cf.~Definition 2). Let $\mathcal{B}$ be a Bell operator over those probabilities and let
\begin{equation}
	\mathcal{T} \equiv \mathcal{B} \circ \mathcal{P},
\end{equation}
i.e. a transformation from~\eqref{eq:SigmaMs} to the Bell operator value, cf.~\eqref{eq:T}.

Now, if we provide some bound $\Lambda$ on the spectrum of the state $W$, we may use the Lemma~\ref{lem:Opts} to perform SDP optimization of the Bell operator $\mathcal{B}$ over hybrid boxes. The levels of the hierarchy are specified by the bound $\Lambda$ of $W$ and the level of the DPS hierarchy used.

\section{The case with two no-signaling and one quantum party}
\label{sec:app1q2ns}

We illustrate the results of Appendix~\ref{sec:Hierarchy} for the case with two no-signaling and one quantum party. For the sake of simplicity we denote in this section the no-signaling parties by $A$ and $B$, and the quantum party by $C$. Moreover, we assume all the parties use binary settings and outcomes - in this case we may use the previous discussion for outer approximation of $\textbf{HNSQ}(2+1,2,2)$, since according to Remark 4 in the main paper the set of hybrid no-signaling-quantum correlations is a convex hull of correlations obtained by projective measurements on no-signaling assemblages acting only on 2-dimensional space.

As shown above in the construction of operator $W$ from Theorem 1 (see Appendix \ref{appA} and Appendix \ref{appB}) for the 3-partite scenario in the dimension $2$ with no-signaling $A$ and $B$, and quantum $C$, in fact for measurements only optimization over second measurement on $C$ is necessary, since the rest of them is taken in either computational or Hadamard basis, and thus in~\eqref{eq:WM} we may assume that only one additional subspace $C''$ is needed.

Indeed, let $M^{a|x}_{A'}$ and $M^{b|y}_{B'}$ be measurements in computational ($x,y=0$) and Hadamard basis ($x,y=1$), for Alice and Boba, that are sufficinet as show in Theorem 1. Using~\eqref{eq:Werner89} the Charlie's measurements can be expressed using:
\begin{subequations}
	\label{eq:measurements}
	\begin{equation}
		\label{eq:A6a}
		M^{c=0|z=0}_{C' C''} \equiv \SWAP(C', C''),
	\end{equation}
	\begin{equation}
		\label{eq:A6b}
		M^{c=0|z=1}_{C' C''} \equiv \kb{0}{0}_{C'} \otimes \openone_{C''},
	\end{equation}
\end{subequations}
and for $c=1$ obtained by summing to identity.

Note that the construction~\eqref{eq:A6a} works for any number of parties and outcomes, whereas~\eqref{eq:A6b} works only in this particular case.

From the bound provided according to construction given in Theorem 1 (see Appendix~\ref{appB}) with a choice of basis dual to $\left\{\openone, |0\rangle \langle 0|, |+\rangle \langle +| \right\}$ it follows that in order to cover all cases allowed by theory in this case we may bound the eigenvalues of $W$ by $\Lambda = 11$, as $\left\|W\right\|\leq (f(\theta_{+}))^2\leq 11$.


Let us now consider the following state where on space $ABC$ there is an additional maximally mixed stated:
\begin{equation}
	\label{eq:EWM}
	\begin{aligned}
		&\Xi_{ABCA'B'C'C''} \equiv \frac{1}{8} \openone_{ABC} \otimes \xi_{A'B'C'C''} = \\
		&\sum_i \lambda_i \left[ \frac{1}{8} \openone_{ABC} \otimes \kb{\phi_i}{\phi_i}_{A'B'C'} \tilde{\otimes} \kb{\psi_i}{\psi_i}_{C''} \right],
	\end{aligned}
\end{equation}
where $\sum_i \lambda_i = 1$, $\lambda_i \geq 0$.


Following the Lemma~\ref{lem:Opts} we get that any product of allowed operator $W$ with the state representing the second measurement for $C$ can be expressed in the following form:
\begin{equation}
	\begin{aligned}
		(1+\Lambda d)& \Tr_{ABC} \left[ \Xi_{ABCA'B'C'C''} \right] \\
		&- \Lambda d \Tr_{A'B'C'} \left[ \Xi_{ABCA'B'C'C''} \right]
	\end{aligned}
\end{equation}

Now, applying the measurements discussed above~\eqref{eq:measurements}, we see that
\begin{equation}
	\label{eq:sigmaM}
	\begin{aligned}
		\Sigma_{ab|xy} \equiv& (1+\Lambda d) \Tr_{ABCA'B'} \left[ \varXi \mu^{ab|xy}_{A'B'} \right] \\
		&- \Lambda d \Tr_{ABA'B'C'} \left[ \varXi \mu^{ab|xy}_{AB} \right],
	\end{aligned}
\end{equation}
where
\begin{subequations}
	\begin{equation}
		\mu^{ab|xy}_{A'B'} \equiv \openone_{ABC} \otimes M^{a|x}_{A'} \otimes M^{b|y}_{B'} \otimes \openone_{C' C''},
	\end{equation}
	\begin{equation}
		\mu^{ab|xy}_{AB} \equiv M^{a|x}_{A} \otimes M^{b|y}_{B} \otimes \openone_{CA'B'C'C''},
	\end{equation}
\end{subequations}
provides us the joint state of the steered states of $C$, $\sigma_{ab|xy}$, and the $C$'s measurement, i.e.
\begin{equation}
	\Sigma_{ab|xy} = \sigma_{ab|xy} \otimes \kb{\psi_i}{\psi_i}.
\end{equation}

From the main text it follows that we should also impose a constraint $\sigma_{ab|xy} \succeq 0$, or, equivalently, $\Sigma_{ab|xy} \succeq 0$. This is expressed by the constraint $z \succeq 0$ in~\eqref{eq:Opt1} and~\eqref{eq:Opt2}. From the above and~\eqref{eq:Werner89} it follows that the probabilities are given by
\begin{equation}
	p(a,b,c|x,y,z) \equiv \Tr \left[ \Sigma_{ab|xy} M^{c|z} \right],
\end{equation}
that definies the function $\mathcal{P}$ from the Appendix~\ref{sec:Hierarchy}.

We have observed that it is beneficial to impose an additional PPT constraint on~\eqref{eq:sigmaM} to strenghten the separation relaxation between the state and measurement operators. The same can be obtained directly by considering a higher level of the DPS relaxation, but our observation suggest that the direct constraining of~\eqref{eq:sigmaM} provides significant improvements at relatively low computational cost.


\begin{thebibliography}{99}
\bibitem{Bell} J. S. Bell, Physics \textbf{1}, 195-200 (1964). 
\bibitem{RMPBellnonlocality} N. Brunner, D. Cavalcanti, S. Pironio, V. Scarani, S. Wehner, Rev. Mod. Phys. \textbf{86}, 419 (2014).
\bibitem{RMPBuhrman} H. Buhrman, R. Cleve, S.Massar, R. De Wolf, Rev. Mod. Phys. \textbf{82}, 665 (2010).
\bibitem{PironioNature} S. Pironio, A. Ac\'in, S. Massar, A. B. de La Giroday, D. N. Matsukevich, P. Maunz, S. Olmschenk, D. Hayes, L. Luo, T. A. Manning, C. Monroe, Nature, \textbf{464}, 1021-1024 (2010). 
\bibitem{Barrett2005}J. Barrett, S. Pironio, S. Popescu, D. Rohrlich, Phys. Rev. Lett. \textbf{95}, 140401 (2005).
\bibitem{Pirandola2019} S. Pirandola et al, Advances in Quantum Cryptography, arXiv:1906.01645 (2019). 
\bibitem{PRLAcin} A. Ac\'in, N. Gisin, L. Masanes, Phys. Rev. Lett. \textbf{97}, 120405 (2006).
\bibitem{IEEEKessler} M. Kessler and R. Arnon-Friedman, IEEE Journal on Selected Areas in Information Theory, \textbf{1}, 568-584. (2017).
\bibitem{BRGH+16} F. G. S. L. Brand{\~a}o, R. Ramanathan, A. Grudka, K. Horodecki, M. Horodecki, P. Horodecki, T. Szarek, H. Wojew{\'o}dka. \textit{Realistic noise-tolerant randomness amplification using finite number of devices}. Nat. Commun. \textbf{7}, 11345 (2016).
\bibitem{RBHH+16} R. Ramanathan, F. G. S. L. Brand{\~a}o, K. Horodecki, M. Horodecki, P. Horodecki, and H. Wojew{\'o}dka. \textit{Randomness Amplification under Minimal Fundamental Assumptions on the Devices}. Phys. Rev. Lett. \textbf{117}, 230501 (2016). 
\bibitem{Rohlich-Popescu}S. Popescu, D. Rohrlich, Found. Phys. \textbf{24}, 379 (1994).
\bibitem{HR19} P. Horodecki and R. Ramanathan. Nat. Comm. \textbf{10}, 1701 (2019).
\bibitem{BBBEW10} H. Barnum, S. Beigi, S. Boixo, M. B. Elliott, S. Wehner. Phys. Rev. Lett. \textbf{104}, 140401 (2010).
\bibitem{S36} E. Schr{\"o}dinger, \textit{Mathematical Proceedings of the Cambridge Philosophical Society} \textbf{32}, 446-452 (1936).
\bibitem{WJD07} H. M. Wiseman, S. J. Jones, A. C. Doherty, Phys. Rev. Lett. \textbf{98} 140402 (2007).
\bibitem{SBCSV15} A. B. Sainz, N. Brunner, D. Cavalcanti, P. Skrzypczyk, T. V{\'e}rtesi, 
Phys. Rev. Lett. \textbf{115}, 190403 (2015).
\bibitem{SAPHS18}A. B. Sainz, L. Aolita, M. Piani, M. J. Hoban, P. Skrzypczyk, New J. Phys. \textbf{20}, 083040 (2018). 
\bibitem{SHSA20} A. B. Sainz, M. J. Hoban, P. Skrzypczyk, L. Aolita, Phys. Rev. Lett. \textbf{125}, 050404 (2020).
\bibitem{HS18}M. J. Hoban, A. B. Sainz, New J. Phys. \textbf{20}, 053048 (2018).
\bibitem{CS17} D. Cavalcanti, P. Skrzypczyk, Rep. Prog. Phys. \textbf{80}, 024001 (2017).
\bibitem{UCNG20} R. Uola, A. C. S. Costa, H. C. Nguyen, O. G{\"u}hne,	Rev. Mod. Phys. \textbf{92}, 15001 (2020).
\bibitem{CS15} D. Cavalcanti, P. Skrzypczyk, G. H. Aguilar, R. V. Nervy, P. H. Souto Riberio, S. P. Walborn, Nat. Commun. \textbf{6}, 7941 (2015).
\bibitem{AQ}
M. Navascu\'{e}s, Y. Guryanova, M. J. Hoban and A. Ac\'{i}n,
Nat. Comm. 6, 6288 (2015).
\bibitem{IC} M. Paw{\l}owski, T. Paterek, D. Kaszlikowski, V. Scarani, A. Winter and M. Zukowski. Nature \textbf{461}, 1101 (2009).
\bibitem{G89} N. Gisin, Helvetica Physica Acta \textbf{62}, 363 (1989).
\bibitem{HJW93} L. P. Hughston, R. Jozsa, K. Wooters, Phys. Lett. A \textbf{183}, 14 (1993).
\bibitem{SNC04}P. Skrzypczyk, M. Navascu{\'e}s, D. Cavalcanti, Phys. Rev. Lett. \textbf{112}, 180404 (2014).
\bibitem{RTHHPRL} R. Ramanathan, J. Tuziemski, M. Horodecki, P. Horodecki, Phys. Rev. Lett. \textbf{117}, 050401 (2016).
\bibitem{RBRH} R. Ramanathan, M. Banacki, R. R. Rodr{\'i}guez, P. Horodecki, (2020), arXiv:2004.14782.
\bibitem{BRH21} M. Banacki, R. R. Rodr{\'i}guez, P. Horodecki, Phys. Rev. A \textbf{103}, 052434 (2021).
\bibitem{Witworld} P. J. Cavalcanti, J. H. Selby, J. Sikora, T. D. Galley, A. B. Sainz, (2021), arXiv:2102.06581.
\bibitem{DPS}A. C. Doherty, P. A. Parrillo, F. M. Spedalieri, Phys. Rev. A \textbf{69}, 022308 (2004).
\bibitem{DPS2} M. Navascu{\'e}s, G. de la Torre, T. V{\'e}rtesi, Phys. Rev. X \textbf{4}, 011011 (2014).
\bibitem{DPS3} A. C. Doherty, P. A. Parrillo, F. M. Spedalieri, Phys. Rev. Lett. \textbf{88}, 187904 (2002).
\bibitem{Werner89} R. F. Werner, Phys. Rev. A \textbf{40}, 4277 (1989).
\bibitem{Sliwa}C. \'{S}liwa, Phys. Lett. A \textbf{317}, 165-168 (2003).
\bibitem{Unified}A. Ac\'{i}n, R. Augusiak, D. Cavalcanti, C. Hadley, J. K. Korbicz, M. Lewenstein, Ll. Masanes, M. Piani, Phys. Rev. Lett. \textbf{104}, 140404 (2010).
\bibitem{RA12} R. Arnon-Friedman and A. Ta-Shma. Phys. Rev. A \textbf{86}, 062333 (2012).
\bibitem{Jordan}C. Jordan, \textit{Essai sur la geometrie a n dimensions}, Bulletin de la S. M. F. 3, 103 (1875).
\bibitem{PV10} K. F. Pal and T. V\'{e}rtesi. Phys. Rev. A \textbf{82}, 022116 (2010).
\bibitem{Liang1}S.-L. Chen, C. Budroni, Y.-C. Liang, Y.-N. Chen, Phys. Rev. Lett. \textbf{116}, 240401 (2016).
\bibitem{Liang2}S.-L. Chen, C. Budroni, Y.-C. Liang, Y.-N. Chen, Phys. Rev. A \textbf{98}, 042127 (2018).
\bibitem{mathematica} Wolfram Research, Inc., \textit{Mathematica, Version 12.3.1}, Champaign, IL (2021).
\bibitem{Supp}Supplemental Material.
\end{thebibliography}
\end{document}